\documentclass[aps,prd,12pt,superscriptaddress,preprint,floatfix]{revtex4-1}
\usepackage[dvips]{graphicx}
\usepackage{amsmath,amssymb,amsthm,bm,multirow,longtable}
\newcommand{\nn}{\nonumber}
\renewcommand{\epsilon}{\varepsilon}
\renewcommand{\phi}{\varphi}

\newcommand{\R}{\mathbb{R}}
\renewcommand{\S}{\mathbb{S}}
\newcommand{\Z}{\mathbb{Z}}

\DeclareMathOperator{\tr}{tr}

\DeclareMathOperator{\curl}{curl}
\providecommand{\abs}[1]{\left\lvert#1\right\rvert}
\providecommand{\norm}[1]{\lVert#1\rVert}
\theoremstyle{plain}

\newtheorem{theorem}{Theorem}
\theoremstyle{definition}
\allowdisplaybreaks

\begin{document}
\title{Generalized helicity and Beltrami fields}
\author{Roman V. Buniy}
\email{roman.buniy@gmail.com}
\affiliation{Schmid College of Science, Chapman University, Orange, CA 92866, USA}
\affiliation{Isaac Newton Institute, University of Cambridge, Cambridge, CB3 0EH, United Kingdom}
\author{Thomas W. Kephart}
\email{tom.kephart@gmail.com}
\affiliation{Department of Physics and Astronomy, Vanderbilt University, Nashville, TN 37235, USA}
\affiliation{Isaac Newton Institute, University of Cambridge, Cambridge, CB3 0EH, United Kingdom}
\date{\today}
\begin{abstract}
We propose covariant and non-abelian generalizations of the magnetic helicity and Beltrami equation. 
The gauge invariance, variational principle, conserved current, energy-momentum tensor and choice of boundary conditions elucidate the subject. 
In particular, we prove that any extremal of the Yang-Mills action functional  $\tfrac{1}{4}\int_\Omega\tr{F_{\mu\nu}F^{\mu\nu}}\,d^4x$ subject to the local constraint $\epsilon^{\mu\nu\alpha\beta}\tr{F_{\mu\nu}F_{\alpha\beta}}=0$ satisfies the covariant non-abelian Beltrami equation.
\end{abstract}
\maketitle

\section{Introduction}

The introduction of the concept of magnetic helicity \cite{Woltier} revolutionized our understanding of plasma physics phenomena, from  dynamos \cite{Moffatt} to the solar wind, to the operation of controlled fusion devices \cite{Marsh1996,Bellan2000,Bellan2006}, and it plays a central role when applied to a variety of concepts such as the Beltrami equation and  force-free fields in the form of Taylor states \cite{Taylor}.
Helicity was introduced in a three-dimensional context, but we find it useful to covariantize it \cite{carter,bekenstein}.
We propose and explore a covariantization of the Beltrami equation, and in particular study its relation to helicity conservation.
For instance, new terms can arise that vanish in the non-relativistic case, but which can contribute when we have a multi-component helicity system where the components are moving with relativistic velocities with respect to each other.
This could be the case for helicity in relativistic plasmas \cite{Lichnerowicz} ejected from astrophysical objects colliding with another plasma clouds \cite{Goedbloed}, but it can also apply in particle physics to the hadronization process where relativistic flux tubes interact \cite{Casher:1978wy,Buniy:2002yx} or to the early universe  as it cools through various epochs.
Some of these examples are Yang-Mills systems and they will need non-abelian generalization.

\section{Non-covariant case}

We first briefly review the three-dimensional Beltrami equation and helicity.

Consider a region $\Omega\subset\R^3$ and let $(x,\nabla)$ be the Cartesian coordinates and the derivative operator in $\Omega$.
A vector field $B$ in $\Omega$ is called a Beltrami vector field if it satisfies
\begin{align}
  B\times(\nabla\times B)=0,
  \label{b3def}
\end{align}
where $\times$ is the vector product in $\R^3$.
Eq.~\eqref{b3def} implies that the vectors $\nabla\times B$ and $B$ are parallel, which leads to the Beltrami equation
\begin{align}
  \nabla\times B=\lambda B,
  \label{b3}
\end{align}
where $\lambda$ is a scalar function in $\Omega$.

Eq. \eqref{b3} shows that Beltrami fields are eigenfields of the curl operator. 
Eigenfields of the curl operator can be related to more familiar functions by using the identity
\begin{align}
  \nabla\times\nabla\times B=\nabla(\nabla\cdot B)-\Delta B,
  \label{}
\end{align}
where $\cdot$ is the scalar product and $\Delta$ is the Laplace operator in $\R^3$.
It follows that the square of the curl operator, when restricted to the space of divergence-free vector fields, is the negative of the Laplace operator $-\Delta$.
Thus, in some sense, the curl operator is the square root of the operator $-\Delta$ (which itself is a positive operator).

The restriction to the space of divergence-free vector fields is not accidental, but is required by physical considerations of $B$ being a magnetic field.
In such a case, the divergence-free condition $\nabla\cdot B=0$ for $B$ in \eqref{b3} implies 
\begin{align}
  B\cdot\nabla \lambda=0,
  \label{b3lambda}
\end{align}
so that $\lambda$ is constant along any field line of $B$.
Eq. \eqref{b3lambda} is the consistency condition for \eqref{b3}.
According to the Beltrami equation \eqref{b3}, the Maxwell current $J=\nabla\times B$ is parallel to the magnetic field, $J=\lambda B$.
Note that the current conservation $\nabla\cdot J=0$ also implies \eqref{b3lambda}. 

To learn more about a Beltrami field $B$, it is instructive to consider a vector potential $A$ such that $B=\nabla\times A$.
Since a vector potential is defined only up to the gradient of an arbitrary function, it will be important to ensure gauge invariance of various physical quantities under a gauge transformation 
\begin{align}
  &A\mapsto A+\nabla g,\label{gauge3-a}
\end{align}
where $g$ is an arbitrary real-valued function in $\Omega$.
The two simplest such gauge invariant quantities are the energy $W$ and helicity $H$ of the field $B$ in the region $\Omega$,
\begin{align}
  &W=\int_{\Omega}\tfrac{1}{2}\norm{B}^2\,d^3x,\label{e3}\\
  &H=\int_{\Omega}A\cdot B\,d^3x,\label{h3}
\end{align}
where $\norm{\ }$ is the scalar norm in $\R^3$.

Convergence of the integrals in \eqref{e3} and \eqref{h3} imposes certain restrictions on $A$ and $B$.
We are concerned here with the case of a non-compact $\Omega$ and restrictions derived from the required behavior of $A$ and $B$ for $\norm{x}\to\infty$.
In such a case, convergence of the integral in \eqref{h3} implies
\begin{align}
  A=O(\norm{x}^{p}), \ \norm{x}\to\infty, \ p<-1.\label{a3asympt}
\end{align}
This leads to $B=O(\norm{x}^{p-1})$, $\norm{x}\to\infty$, which means that allowed field configurations do not include magnetic monopoles.
It follows that the integral in \eqref{e3} also converges for such fields.

The gauge invariance of the energy is obvious, while the corresponding gauge transformation of the helicity is
\begin{align}
  H\mapsto H+\int_{\partial\Omega}g(n\cdot B)d^2\sigma,
  \label{h3transf}
\end{align}
where $\partial\Omega$ is the boundary of $\Omega$, $n$ is the unit normal vector to $\partial\Omega$, and $d^2\sigma$ is the area differential on $\partial\Omega$.
Since we require gauge invariance of $H$, we set the boundary condition
\begin{align}
  (n\cdot B)\vert_{\partial\Omega}=0,
  \label{bc3}
\end{align}
which means that the field lines do not cross the boundary.
For a non-compact $\Omega$, the asymptotic behavior $B=O(\norm{x}^{p-1})$, $\norm{x}\to\infty$ ensures the gauge invariance of $H$ as well since the boundary integral in \eqref{h3transf} vanishes.

The helicity is often conserved in physical systems involving magnetic fields, and this restricts their dynamics. 
For example, suppose that $B$ is a field in $\Omega$ satisfying the boundary condition \eqref{bc3} which minimizes its energy $W$ and conserves its helicity $H$.
The resulting variational problem is equivalent to finding $B$ minimizing the functional
\begin{align}
  &W-\tfrac{1}{2}\lambda H=\int_{\Omega}L\, d^3 x\label{functional3}
\end{align}
with the Lagrangian
\begin{align}
  &L=\tfrac{1}{2}\norm{B}^2-\tfrac{1}{2}\lambda A\cdot B.\label{lagrangian3}
\end{align}
(We have chosen the form of the Lagrange multiplier $\lambda$ which leads to the conventional form of the Beltrami equation.)
For an infinitesimal variation of the vector potential $\delta A$, we find
\begin{align}
  \delta L=(\nabla\times B-\lambda B)\cdot\delta A+\nabla\cdot\bigl[(-B+\tfrac{1}{2}\lambda A)\times\delta A\bigr],
  \label{var3-l}
\end{align}
which leads to
\begin{align}
  \delta(W-\tfrac{1}{2}\lambda H)=\int_{\Omega}(\nabla\times B-\lambda B)\cdot\delta A\, d^3x+\int_{\partial\Omega} n\cdot\bigl[(-B+\tfrac{1}{2}\lambda A)\times\delta A\bigr]\,d^2\sigma.
  \label{var3}
\end{align}
We eliminate the boundary term in \eqref{var3} by setting the boundary condition
\begin{align}
  \delta A\vert_{\partial\Omega}=0.
  \label{bc3var}
\end{align}
The variation \eqref{var3} vanishes for any $\delta A$ satisfying \eqref{bc3var} if and only if $\nabla\times B=\lambda B$.
Thus, a Beltrami field with $\lambda=\textrm{const}$ is a stationary point of the energy functional $W$ subject to the condition $H=\textrm{const}$.
It can be further proved that such a field is a local minimum of $W$ with constant $H$.

Another aspect of the helicity relates to the conserved Noether current.
Gauge transformations are the symmetry operations of the theory defined by the Lagrangian $L$.
The proof of the invariance of the theory requires showing (without using the equation of motion) that $L$ is changed only by the divergence term.
(For the following derivation we assume $\lambda$ is constant.)
Indeed, for a gauge transformation \eqref{gauge3-a} with $\delta A=\nabla g$, \eqref{var3-l} becomes
\begin{align}
  &\delta L=\nabla\cdot\Gamma,\label{delta-noether-L3} \\
  &\Gamma=-\tfrac{1}{2}\lambda gB.\label{gamma3}
\end{align}
On the other hand, using the equation of motion we find  
\begin{align}
  &\delta L=\nabla\cdot\biggl(\frac{\partial L}{\partial\nabla A_k}\delta A_k\biggr).
  \label{delta-L3}
\end{align}
Equating \eqref{delta-noether-L3} and  \eqref{delta-L3}, we arrive at the conserved Noether current ($\nabla\cdot j=0$),
\begin{align}
  &j=\frac{\partial L}{\partial\nabla A_k}\delta A_k-\Gamma,\label{noether3def}\\
  &j=(-B+\tfrac{1}{2}\lambda A)\times\nabla g+\tfrac{1}{2}\lambda gB.\label{noether3}
\end{align}
Since the gauge function $g$ is arbitrary, we can define another conserved Noether current $k$ by
\begin{align}
  &gk=j+\nabla\times\bigl[(-B+\tfrac{1}{2}\lambda A)g\bigr].\label{}
\end{align}
and find
\begin{align}
  k=-\nabla\times B+\lambda B.
  \label{current3}
\end{align}
Now the Beltrami equation \eqref{b3} gives $k=0$, so that the Noether current associated with the helicity, $\lambda B$, equals the Maxwell current $J=\nabla\times B$.
As expected, there is only one independent conserved current in the problem.

Computing the Noether energy-momentum tensor
\begin{align}
  {\theta^{i}}_j&=\frac{\partial L}{\partial(\nabla_i A_k)}\nabla_j A_k-{\delta^i}_j L\nn\\&=(B_l-\tfrac{1}{2}\lambda A_l)\epsilon^{lik}\nabla_j A_k-{\delta^i}_j\bigl(\tfrac{1}{2}\norm{B}^2-\tfrac{1}{2}\lambda A\cdot B\bigr),
  \label{}
\end{align}
we find that its divergence
\begin{align}
  \nabla_i{\theta^{i}}_j&=(-\nabla\times B+\lambda B)^k\nabla_j A_k
  \label{}
\end{align}
vanishes for any solution of the Beltrami equation \eqref{b3}, which implies the conservation equation  $\nabla_i{\theta^{i}}_j=0$.

We now derive the lower bound for the energy in terms of helicity and constant $\lambda$ \cite{arnold-khesin}.
We first integrate the Beltrami equation \eqref{b3} once to obtain 
\begin{align}
  \nabla\times A=\lambda A+\nabla\phi,
  \label{b3sol}
\end{align}
where $\phi$ is an arbitrary scalar function in $\Omega$.
We can now use the gauge transformation \eqref{gauge3-a} with $g=-\lambda^{-1}\phi$ to replace \eqref{b3sol} with
\begin{align}
  \nabla\times A=\lambda A,
  \label{a3}
\end{align}
which is of the same form as \eqref{b3}.
Hence in this gauge $B=\lambda A$, which gives
\begin{align}
  W=\tfrac{1}{2}\lambda H.
  \label{wh3}
\end{align}
Although \eqref{a3} is not gauge invariant, its consequence, \eqref{wh3}, is gauge invariant. 
We conclude that the minimal value of the variational functional $W-\tfrac{1}{2}\lambda H$ equals zero for any solution of the Beltrami equation with constant $\lambda$.

The field satisfying $B=\lambda A$ saturates the lower bound for the energy in terms of helicity \cite{arnold-khesin}.
To derive this, we consider a non-local operator $\curl^{-1}$ acting on the space of divergence-free vector fields.
We use the Schwarz inequality
\begin{align}
  \biggl\vert\int_{\Omega} B\cdot\curl^{-1}{B}\, d^3 x\biggr\vert\le\biggl[\int_{\Omega}\norm{B}^2\, d^3 x\biggr]^{1/2}\biggl[\int_{\Omega}\norm{\curl^{-1}{B}}^2\, d^3 x\biggr]^{1/2}
  \label{}
\end{align}
and the Poincar\'{e} inequality
\begin{align}
  \int_{\Omega}\norm{\curl^{-1}{B}}^2\, d^3 x\le C^{-2}\int_{\Omega}\norm{B}^2\, d^3 x,
  \label{}
\end{align}
where $C>0$ is a certain constant depending on $\Omega$.
Combination of the two inequalities gives $W\ge\tfrac{1}{2} C\abs{H}$.
Finally, using the Rayleigh min-max theorem
\begin{align}
  &B\cdot\curl^{-1}{B}\le\abs{\mu}_\textrm{max}\norm{B}^2,
  \label{}
\end{align}
where
\begin{align}
  &\abs{\mu}_\textrm{max}=\max_a{\abs{\mu_a}},\\
  &\curl^{-1}B_a=\mu_a B_a,
  \label{}
\end{align}
we see that we can use $C=2\abs{\mu}^{-1}_\textrm{max}$ and find
\begin{align}
  W\ge\tfrac{1}{2}\abs{\mu}^{-1}_\textrm{max}\abs{H}.
  \label{bound3}
\end{align}
It is clear that the field satisfying $B=\lambda A$ saturates the bound \eqref{bound3} since in this case we have $\abs{\mu}^{-1}_\textrm{max}=\abs{\lambda}$ and $W=\frac{1}{2}\abs{\lambda}\abs{H}$.

\section{Covariant case}

The proceeding non-covariant analysis is sufficient for the description of magnetic fields in nonrelativistic plasmas.
Generalizations to the electric case have been carried out \cite{Ehelicity} and applied \cite{Eapp}, but relativistic plasmas require a full covariant analysis.
In particular, this applies to Beltrami fields and helicity.

To deriving the covariant forms of equations obtained in the preceding section, we consider Lorentzian $(\R^{1,3},\Omega,x,\nabla)$, where $\R^{1,3}$ has a constant pseudo-Riemannian metric with signature $(1,3)$.
The magnetic field $B$ is now a part of the gauge field strength tensor $F$.
Using
\begin{align}
  B_i=\tfrac{1}{2}\epsilon_{ijk}F^{jk},
  \label{bf3def}
\end{align}
we first write \eqref{b3def}, \eqref{b3}, \eqref{b3lambda} in the form
\begin{align}
  &{F_i}^j\nabla^{k} F_{jk}=0,\label{b3deff}\\
  &\nabla^{j} F_{ij}=\tfrac{1}{2}\lambda\epsilon_{ijk}F^{jk},\label{b3f}\\
  &\epsilon^{ijk}F_{jk}\nabla_i \lambda=0,\label{b3lambdaf}
\end{align}
then setting
\begin{align}
  &E_i=F_{i0},\label{be3def}\\
  &\epsilon_{0ijk}=\epsilon_{ijk},\label{eps4def}\\
  &\lambda_0=\lambda\label{lambda4def},
\end{align}
we arrive at the covariant form of \eqref{b3deff}, \eqref{b3f}, \eqref{b3lambdaf},
\begin{align}
  &{F_\alpha}^\mu\nabla^\nu F_{\mu\nu}=0,\label{b4def}\\
  &\nabla^\nu F_{\mu\nu}=\tfrac{1}{2}\epsilon_{\mu\nu\alpha\beta}\lambda^\nu F^{\alpha\beta},\label{b4}\\
  &\epsilon^{\mu\nu\alpha\beta}F_{\alpha\beta}\nabla_\mu \lambda_\nu=0.\label{b4lambda}
\end{align}
Covariantization requires that we identify non-covariant $\lambda$ with the time component of a $4$-vector $\lambda$.
Note that the left-hand side of \eqref{b4lambda} vanishes identically if we set
\begin{align}
  \nabla_\mu\lambda_\nu-\nabla_\nu\lambda_\mu=0,
  \label{lambda4}
\end{align}
The requirement \eqref{lambda4} will appear later in the variational formulation of the problem.

Similarly to \eqref{b3} implying \eqref{b3lambda} and \eqref{b3lambda} not implying \eqref{b3} for $\nabla\cdot B=0$, we have \eqref{b4} implying \eqref{b4lambda} and \eqref{b4lambda} not implying \eqref{b4}.
However, although \eqref{b3def} and \eqref{b3} are equivalent, their covariant counterparts \eqref{b4def} and \eqref{b4} are not equivalent; in fact, none of the two implies the other.

In terms of the $E$ and $B$ fields, the time and space components of \eqref{b4def} become
\begin{align}
  &E\cdot \nabla_0 E-E\cdot(\nabla\times B)=0,\label{b4def1}\\
  &E(\nabla\cdot E)+B\times(\nabla_0 E)-B\times(\nabla\times B)=0,\label{b4def2}
\end{align}
the time and space components of \eqref{b4} become
\begin{align}
  &-\nabla\cdot E=\lambda\cdot B,\label{b41}\\
  &-\nabla_0 E+\nabla\times B=\lambda_0 B+\lambda\times E,\label{b42}
\end{align}
and \eqref{b4lambda} becomes
\begin{align}
  &-B\cdot\nabla \lambda_0+B\cdot\nabla_0 \lambda+E\times(\nabla\times \lambda)=0.\label{b4lambda1}
\end{align}
Note that the left-hand side of \eqref{b4lambda1} vanishes identically if we set $\nabla\lambda_0-\nabla_0\lambda=0$ and $\nabla\times\lambda=0$, which combine to give \eqref{lambda4}.

A consistency condition is required for compatibility of \eqref{b4def} and \eqref{b4} for arbitrary $\lambda$.
Indeed, combining these equations, we find
\begin{align}
  &F^{\gamma\mu}\epsilon_{\mu\nu\alpha\beta}\lambda^\nu F^{\alpha\beta}=0.\label{f-eps-lambda-f}
\end{align} 
Considering the values $\gamma=0$ and $\gamma=i$ in \eqref{f-eps-lambda-f}, yields $\lambda^0 E\cdot B=0$ and $\lambda^i E\cdot B=0$, respectively.
This requires the same consistency condition $E\cdot B=0$ for each values of $\gamma$, which we write in the covariant form
\begin{align}
  \epsilon^{\mu\nu\alpha\beta}F_{\mu\nu}F_{\alpha\beta}=0.
  \label{epsff}
\end{align}
We could have arrived at the consistency condition $E\cdot B=0$ also by noting that it is an appropriate covariant form of the three-dimensional constraint $E=0$. 

The covariant analogue of the energy $W$ is the negative of the Maxwell action
\begin{align}
  W&=\int_{\Omega}\tfrac{1}{4}F_{\mu\nu}F^{\mu\nu} d^4x,\nn\\
  &=\int_{\Omega}\left(-\tfrac{1}{2}\norm{E}^2+\tfrac{1}{2}\norm{B}^2\right)d^4x,\label{e4}
\end{align}
(We have introduced the sign difference in the definition of $W$ so that the non-covariant $W$ is a limiting case of the covariant $W$.)
As a covariant form of the helicity $H$, we propose
\begin{align}
  H(f)&=-\int_\Omega\tfrac{1}{2}\epsilon^{\mu\nu\alpha\beta}(\nabla_\mu f)A_\nu F_{\alpha\beta}\,d^4x\nn\\
  &=\int_{\Omega}\tfrac{1}{2}\bigl[(\nabla_0 f)(A\cdot B)-A_0(B\cdot\nabla f)-\nabla f\cdot(A\times E)\bigr]d^4x,
  \label{h4}
\end{align}
where $f$ is an arbitrary scalar function in $\Omega$.
(It will become clear in what follows why in \eqref{h4} we use $\lambda_\mu=\nabla_\mu f$ instead of a general $\lambda_\mu$.)

For a non-compact $\Omega$, convergence of the integral in \eqref{h4} implies
\begin{align}
  &A=O(\norm{x}^p), \ \norm{x}\to\infty, \ p<-1-\tfrac{1}{2}q,\label{a4asympt}
\end{align}
where we assumed
\begin{align}
  &f=O(\norm{x}^q), \ \norm{x}\to\infty\label{f4asympt}
\end{align}
for a certain $q$.
Since $F=O(\norm{x}^{p-1})$, $\norm{x}\to\infty$, convergence of the integral in \eqref{e4} now implies $p<-1$.

Our definition \eqref{h4} is motivated by the following limiting case of covariant helivity $H(f)$. 
Suppose $\Omega=[t_1,t_2]\times\Omega'$, where $\Omega'\subset\R^3$, and $f$ is a function of $x^0=t$ only.
It follows that
\begin{align}
  H(f)=\int_{t_1}^{t_2}H'(t)(\partial f/\partial t)dt,
  \label{}
\end{align}
where $H'(t)$ is the non-covariant helicity of the vector potential $A_i(t,x)$. 
In particular, for the conserved non-covariant helicity $H'$, we find
\begin{align}
  H(f)=\left[f(t_2)-f(t_1)\right]H',
  \label{}
\end{align}
More generally, for an arbitrary $f$, \eqref{h4} implies
\begin{align}
  &H(f)=\tilde{H}(f)-\int_{\partial\Omega}\tfrac{1}{2}\epsilon^{\mu\nu\alpha\beta}fA_\nu F_{\alpha\beta}\,d^3\sigma_\mu,\label{htildeh4}\\
  &\tilde{H}(f)=\int_\Omega \tfrac{1}{4}\epsilon^{\mu\nu\alpha\beta}fF_{\mu\nu}F_{\alpha\beta}\,d^4 x,\label{tildeh4}
\end{align}
which means that $H(f)$ is a boundary term when consistency condition \eqref{epsff} is satisfied.

Under a gauge transformation
\begin{align}
  &A_\mu\mapsto A_\mu+\nabla_\mu g,\label{gauge4-a}
\end{align}
where $g$ is an arbitrary real-valued function in $\Omega$, the gauge invariance of $W$ is obvious, while the corresponding gauge transformation of the helicity is
\begin{align}
  H(f)\mapsto H(f)+\int_{\partial\Omega}\tfrac{1}{2}\epsilon^{\mu\nu\alpha\beta}(\nabla_\nu f)gF_{\alpha\beta} \,d^3\sigma_\mu.
  \label{h4transf}
\end{align}
Since we require gauge invariance of $H(f)$, we set
\begin{align}
  \bigl[\epsilon^{\mu\nu\alpha\beta}n_\mu(\nabla_\nu f)F_{\alpha\beta}\bigr]_{\partial\Omega}=0,
  \label{bc4}
\end{align}
where $n$ is the $4$-vector normal to $\partial\Omega$.
Using now \eqref{b4} with $\lambda_\mu=\nabla_\mu f$, we find
\begin{align}
  \left(n^\mu\nabla^\nu F_{\mu\nu}\right)_{\partial\Omega}=0,
  \label{bc41}
\end{align}
which a covariant version of the boundary condition \eqref{bc3}. 
For a non-compact $\Omega$, the asymptotic behavior $F=O(\norm{x}^{p-1})$, $\norm{x}\to\infty$ ensures the gauge invariance of $H(f)$ as well since the boundary integral in \eqref{h4transf} vanishes.

In terms of the $E$ and $B$ fields, the boundary condition \eqref{bc41} becomes
\begin{align}
  &\left[-n_0(\nabla f\cdot B)+(\nabla_0 f)(n\cdot B)+n\cdot(\nabla f\times E)\right]_{\partial\Omega}=0.
  \label{bc42}
\end{align}
In particular, for time-like and space-like hypersurface $\partial\Omega$ we have 
\begin{align}
  &(\nabla f\cdot B)_{\partial\Omega}=0 \textrm{\ for time-like\ } \partial\Omega,\\
  &\left[(\nabla_0 f)(n\cdot B)+n\cdot(\nabla f\times E)\right]_{\partial\Omega}=0 \textrm{\ for space-like\ } \partial\Omega.
  \label{}
\end{align}

We further emphasize the choice of the definition \eqref{h4} by the following theorem.
\begin{theorem}
Any extremal of the action functional $W=\int_{\Omega}\tfrac{1}{4}F_{\mu\nu}F^{\mu\nu} d^4x$ subject to the constraint $\epsilon^{\mu\nu\alpha\beta}F_{\mu\nu}F_{\alpha\beta}=0$ and the boundary condition $\delta A\vert_{\partial\Omega}=0$ satisfies the covariant Beltrami equation $\nabla^\nu F_{\mu\nu}=\tfrac{1}{2}\epsilon_{\mu\nu\alpha\beta}\lambda^\nu F^{\alpha\beta}$ for $\lambda_\mu=\nabla_\mu f$, where $f=f(x)$ is an arbitrary function.
  \label{theorem1}
\end{theorem}
\begin{proof}
Any extremal of \eqref{e4} subject to the local constraint \eqref{epsff} must be an extremal of the functional $W-\tfrac{1}{2}\tilde{H}(f)$, where $f(x)$ is a space-time dependent Lagrange multiplier~\cite{vc}.
For an arbitrary variation of the gauge potential $\delta A$, we find
\begin{align}
  \delta(W-\tfrac{1}{2}\tilde{H}(f))&=\int_\Omega\bigl[-(\nabla_\mu F^{\mu\nu})+\tfrac{1}{2}\epsilon^{\mu\nu\alpha\beta}(\nabla_\mu f) F_{\alpha\beta}\bigr]\delta A_\nu\,d^4x\nn\\ &+\int_{\partial\Omega}(F^{\mu\nu}-\tfrac{1}{2}\epsilon^{\mu\nu\alpha\beta}f F_{\alpha\beta})\delta A_\nu \,d^3\sigma_\mu.
  \label{var4}
\end{align}
Using the boundary condition
\begin{align}
  \delta A\vert_{\partial\Omega}=0,
  \label{bc4var}
\end{align}
we arrive at \eqref{b4} with $\lambda_\mu=\nabla_\mu f$, which proves the theorem.
\end{proof}

Note that $\lambda_\mu=\nabla_\mu f$ derived in the proof implies \eqref{lambda4}, which we have already seen as a sufficient condition for \eqref{b4lambda} to be satisfied identically.

We now consider the covariant version of the conserved Noether current \cite{Jackiw}.
Equations \eqref{functional3}, \eqref{lagrangian3}, \eqref{delta-noether-L3}, \eqref{noether3def}, \eqref{noether3}, \eqref{current3} become 
\begin{align}
  &W-\tfrac{1}{2}\tilde{H}(f)=\int_\Omega L(f)\, d^4 x,\label{functional4} \\
  &L(f)=\tfrac{1}{4}F_{\mu\nu}F^{\mu\nu}-\tfrac{1}{8}\epsilon^{\mu\nu\alpha\beta}fF_{\mu\nu}F_{\alpha\beta},\label{lagrangian4} \\
  &\delta L(f)=0,\label{delta-noether4} \\
  &j^\mu(f)=\frac{\partial L(f)}{\partial\nabla_\mu A_\nu}\delta A_\nu, \label{noether4def} \\
  &j^\mu(f)=\left(F^{\mu\nu}-\tfrac{1}{2}\epsilon^{\mu\nu\alpha\beta}fF_{\alpha\beta}\right)\nabla_\nu g, \label{noether4} \\
  &k^\mu(f)=-\nabla_\nu F^{\mu\nu}+\tfrac{1}{2}\epsilon^{\mu\nu\alpha\beta}(\nabla_\nu f)F_{\alpha\beta}. \label{current4}
\end{align}
Now using the Beltrami equation \eqref{b4}, we find $k^\mu(f)=0$, so that the Noether current associated with the helicity, $\tfrac{1}{2}\epsilon^{\mu\nu\alpha\beta}(\nabla_\nu f)F_{\alpha\beta}$, equals the negative of the Maxwell current $J^\mu=-\nabla_\nu F^{\mu\nu}$.
As expected, there is only one independent conserved current in the problem.

The Noether energy-momentum tensor is
\begin{align}
  {\theta^\mu}_\nu(f)&=\frac{\partial L(f)}{\partial(\nabla_\mu A_\sigma)}\nabla_\nu A_\sigma-{\delta^\mu}_\nu L(f),\nn\\
  &=(F^{\mu\sigma}-\tfrac{1}{2}f\epsilon^{\mu\sigma\alpha\beta}F_{\alpha\beta})\nabla_\nu A_\sigma-{\delta^\mu}_\nu(\tfrac{1}{4}F_{\alpha\beta}F^{\alpha\beta}-\tfrac{1}{8}\epsilon^{\alpha\beta\gamma\delta}fF_{\alpha\beta}F_{\gamma\delta})
  \label{theta4}
\end{align}
and the corresponding energy-momentum $4$-vector is
\begin{align}
  &P_\nu(f)=\int_{\Omega'}{\theta^0}_\nu(f)\,d^3 x,\\
  &P_0(f)=-\int_{\Omega'}\bigl(\tfrac{1}{2}\norm{E}^2+\tfrac{1}{2}\norm{B}^2\bigr)\,d^3 x+\int_{\partial\Omega'}n\cdot(E+fB) A_0 \,d^2\sigma,\\
  &P_i(f)=-\int_{\Omega'}(E\times B)_i\, d^3 x+\int_{\partial\Omega'}n\cdot(E+fB)A_i \,d^2\sigma,
  \label{}
\end{align}
where we assumed $\Omega=[t_1,t_2]\times\Omega'$, $\Omega'\subset\R^3$.
We set the boundary condition
\begin{align}
  n\cdot(E+fB)\vert_{\partial\Omega'}=0
  \label{bc-p4}
\end{align}
and obtain the relation $P_\nu(f)=P_\nu(0)$ which is consitent with $L(f)-L(0)$ being a topological term.
Also note that although ${\theta^\mu}_\nu(f)$ is not gauge invariant, the resulting $P_\nu(f)$ is.

To prove conservation of ${\theta^\mu}_\nu(f)$, we need to use the Beltrami equation.
Indeed, in the expression 
\begin{align}
  \nabla_\mu{\theta^\mu}_\nu(f)&=\bigl((\nabla_\mu F^{\mu\sigma})-\tfrac{1}{2}\epsilon^{\mu\sigma\alpha\beta}(\nabla_\mu f)F_{\alpha\beta}\bigr)\nabla_\nu A_\sigma+\tfrac{1}{8}\epsilon^{\alpha\beta\gamma\delta}(\nabla_\nu f)F_{\alpha\beta}F_{\gamma\delta},
  \label{nabla-theta4}
\end{align}
the first term on the right-hand side vanishes for any solution of \eqref{b4} and the second term vanishes due to the constraint \eqref{epsff}. 
Since \eqref{epsff} follows from  \eqref{b4}, we conclude that the conservation equation $\nabla_\mu{\theta^\mu}_\nu(f)=0$ holds for any solution of the Beltrami equation.

\section{Non-abelian case}

So far we have worked with Maxwell's electromagnetism, which is an abelian gauge theory.
We anticipate applications of generalized helicity and Beltrami equation to non-abelian theories as well.
For example, in high temperature QCD with free quarks and gluons, conservation of chromomagnetic and chromoelectric helicity could affect the dynamics by restricting the evolution of configurations as the system cools.
This could apply to a range of situations from the early Universe to high energy nucleus-nucleus collisions at the LHC.
Furthermore, at lower energy per particle the hadronization process will involve chromoelectric fields confined to flux tubes and bags, so chromoelectric helicity conservation should play a role in determining decays and final states.

To proceed, we choose a non-abelian gauge group $G$, its algebraic generators $\{T_a\}$, and the corresponding structure constants $\{e_{abc}\}$ satisfying the commutation relation $[T_a,T_b]=e_{abc}T^c$.
The anti-hermitian generators are normalized according to $\tr{T_a T_b}=-\delta_{ab}$.
The gauge field $A_\mu=A^a_\mu T_a$ and the field strength $F_{\mu\nu}=F^a_{\mu\nu} T_a$ are elements of the algebra of $G$, and are related according to
\begin{align}
  F_{\mu\nu}=\nabla_\mu A_\nu-\nabla_\nu A_\mu+[A_\mu,A_\nu].
  \label{FA}
\end{align}
We also need the gauge covariant derivative of the field strength,
\begin{align}
  D_\alpha F_{\mu\nu}=\nabla_\alpha F_{\mu\nu}+[A_\alpha,F_{\mu\nu}].
  \label{}
\end{align}
Under a gauge transformation
\begin{align}
  &A_\mu\mapsto U^{-1}A_\mu U+U^{-1}\nabla_\mu U,\label{gauge4na-a}
\end{align}
where $U$ is an arbitrary $G$-valued function in $\Omega$, we have
\begin{align}
  &F_{\mu\nu}\mapsto U^{-1}F_{\mu\nu}U,\label{gauge4na-f}\\
  &D_\alpha F_{\mu\nu}\mapsto U^{-1}(D_\alpha F_{\mu\nu})U,\label{gauge4na-df}
\end{align}
  
To generalize the results of the previous section to a non-abelian group $G$, we need to ensure that all equations transform properly under the above gauge transformations.
Such generalizations are straightforward in most cases; for example, equations \eqref{b4def}, \eqref{b4}, \eqref{b4lambda} are replaced with
\begin{align}
  &{F_\alpha}^\mu D^\nu F_{\mu\nu}=0,\label{b4na-def}\\
  &D^\nu F_{\mu\nu}=\tfrac{1}{2}\epsilon_{\mu\nu\alpha\beta}\lambda^\nu F^{\alpha\beta},\label{b4na}\\
  &\epsilon^{\mu\nu\alpha\beta}F_{\alpha\beta}\nabla_\mu \lambda_\nu=0.\label{b4na-lambda}
\end{align}

We need to be careful, however, when generalizing \eqref{epsff}.
To derive the corresponding equation, we first note that \eqref{b4na-def} and \eqref{b4na} lead to
\begin{align}
  &F^{\gamma\mu}\epsilon_{\mu\nu\alpha\beta}\lambda^\nu F^{\alpha\beta}=0.\label{f-eps-lambda-f-na}
\end{align}
Proceeding as in the abelian case by considering the values $\gamma=0$ and $\gamma=i$ separately, we arrive at
\begin{align}
  &\lambda^0 E_i B^i+\epsilon_{ijk}\lambda^i E^j E^k=0,\label{consistency-na1}\\
  &\lambda^j(E^i B_j-B_j E^i)-\lambda^0\epsilon^{ijk}B_j B_k+\lambda^i B_j E^j=0,\label{consistency-na2}
\end{align}
where $G$-valued electric and magnetic fields are
\begin{align}
  &E_i=F_{i0},\\
  &B_i=\tfrac{1}{2}\epsilon_{ijk}F^{jk}.
  \label{}
\end{align}
Due to non-commutativity of the $E$ and $B$ fields, we cannot conclude from \eqref{consistency-na1} and \eqref{consistency-na2} that $E_i B^i=0$.
However, taking the trace of \eqref{consistency-na1} and \eqref{consistency-na2}, we arrive at the non-abelian consistency condition $\tr{E_i B^i}=0$, which we write in the covariant form generalizing \eqref{epsff},
\begin{align}
  \epsilon^{\mu\nu\alpha\beta}\tr{F_{\mu\nu}F_{\alpha\beta}}=0.\label{epsff-na}
\end{align}

Equations \eqref{e4}, \eqref{h4}, \eqref{htildeh4}, \eqref{tildeh4} become
\begin{align}
  &W=\int_{\Omega}\tfrac{1}{4}\tr{F_{\mu\nu}F^{\mu\nu}} d^4x,\label{e4na}\\
  &H(f)=-\int_\Omega\tfrac{1}{2}\epsilon^{\mu\nu\alpha\beta}(\nabla_\mu f)\tr{\bigl(A_\nu F_{\alpha\beta}-\tfrac{2}{3}A_\nu A_\alpha A_\beta\bigr)}\,d^4x,\label{h4na}\\
  &H(f)=\tilde{H}(f)-\int_{\partial\Omega}\tfrac{1}{2}\epsilon^{\mu\nu\alpha\beta}f\tr{\bigl(A_\nu F_{\alpha\beta}-\tfrac{2}{3}A_\nu A_\alpha A_\beta\bigr)}\,d^3\sigma_\mu,\label{htildeh4na}\\
  &\tilde{H}(f)=\int_\Omega\tfrac{1}{4} \epsilon^{\mu\nu\alpha\beta}f\tr{F_{\mu\nu}F_{\alpha\beta}}\,d^4 x,\label{tildeh4na}
\end{align}
where the gauge invariance of $H(f)$ requires the appearance of the well-known term cubic in $A$.

For a non-compact $\Omega$, convergence of the integral in \eqref{h4na} for $\norm{x}\to\infty$ implies
\begin{align}
  &A=O(\norm{x}^p), \ \norm{x}\to\infty, \ p<-1-\tfrac{1}{2}q, \ p<-1-\tfrac{1}{3}q, \label{a4na-asympt}
\end{align}
where we assumed
\begin{align}
  &f=O(\norm{x}^q), \ \norm{x}\to\infty\label{f4na-asympt}
\end{align}
for a certain $q$.
Since $F=O(\norm{x}^{p-1})$, $\norm{x}\to\infty$, convergence of the integral in \eqref{e4na} now implies $p<-1$.

Under the gauge transformation \eqref{gauge4na-a}, the invariance of \eqref{e4na} is obvious, while the corresponding transformation of \eqref{h4na} is
\begin{align}
  H(f)\mapsto H(f)&-\int_{\partial\Omega}\epsilon^{\alpha\beta\mu\nu}f\tr\nabla_\alpha\bigl((\nabla_\mu U)U^{-1} A_\nu\bigr)\,d^3\sigma_\beta\nn\\
  &-\int_{\partial\Omega}\tfrac{1}{3}\epsilon^{\alpha\beta\mu\nu}f\tr U^{-1}(\nabla_\mu U)U^{-1}(\nabla_\nu U)U^{-1}(\nabla_\alpha U)\,d^3\sigma_\beta.\label{h4na-transf}
\end{align}
We note two significant differences between \eqref{h4na-transf} and its abelian counterpart \eqref{h4transf}.

First, vanishing of the first integral in \eqref{h4na-transf} leads to a more restrictive boundary condition than the similar procedure for \eqref{h4transf}.
To see this, we evaluate the first integral in \eqref{h4na-transf} for an infinitesimal transformation with the gauge function 
\begin{align}
  U=\exp{(g)}, \ g\to 0\label{gauge4inf}
\end{align}
and find
\begin{align}
\int_{\partial\Omega}\epsilon^{\alpha\beta\mu\nu}f\tr\nabla_\alpha\bigl((\nabla_\mu U)U^{-1} A_\nu\bigr)\,d^3\sigma_\beta =\int_{\partial\Omega}\tfrac{1}{2}\epsilon^{\alpha\beta\mu\nu}(\nabla_\mu f)\tr g\bigl(F_{\nu\alpha}-[A_\nu,A_\alpha]\bigr)\,d^3\sigma_\beta +O(g^2).
  \label{h4na-transf-inf}
\end{align}
For the $O(g)$ term in \eqref{h4na-transf-inf} to vanish for any $g$, we need to impose the condition
\begin{align}
  \bigl(\epsilon^{\alpha\beta\mu\nu}n_\beta(\nabla_\mu f)(F_{\nu\alpha}-[A_\nu,A_\alpha])\bigr)_{\partial\Omega}=0.\label{bc4na}
\end{align}
Since the quantity $(F_{\nu\alpha}-[A_\nu,A_\alpha])_{\partial\Omega}$ is not gauge invariant, we conclude that \eqref{bc4na} requires
\begin{align}
  (\nabla_\mu f)_{\partial\Omega}=0.
  \label{bc4na1}
\end{align}
Now the Beltrami equation \eqref{b4na} with $\lambda_\mu=\nabla_\mu f$ implies
\begin{align}
  (D^\nu F_{\mu\nu})_{\partial\Omega}=0.
  \label{bc4na2}
\end{align}
Equations \eqref{bc4na1} and \eqref{bc4na2} are the boundary conditions needed for the invariance of $H(f)$ under transformations with the gauge function of the form \eqref{gauge4inf}. 
We see that the nonabelian boundary conditions \eqref{bc4na1} and \eqref{bc4na2} are more restrictive than their abelian counterpart \eqref{bc41} because \eqref{bc4na2} implies 
\begin{align}
  \left(n^\mu D^\nu F_{\mu\nu}\right)_{\partial\Omega}=0,
  \label{bc4na3}
\end{align}
which is the nonabelian generalization of \eqref{bc41}, but \eqref{bc4na3} does not imply either \eqref{bc4na1} or \eqref{bc4na2}.

Transformations with gauge functions that can be written in the exponential form \eqref{gauge4inf} (with $g$ not necessarily small) are called small gauge transformations because they are homotopically equivalent to the identity transformation.
Gauge functions for all other transformations, which are called large gauge transformations, cannot be written in the exponential form \eqref{gauge4inf} and are topologically nontrivial.
This brings us to the second distinction of the nonabelian case, namely, that the analog of the second integral in \eqref{h4na-transf} does not appear in \eqref{h4transf}.
This term depends only on $f$ and the gauge function $U$, and its independence from the gauge field $A$ is significant. 
Further note that the boundary condition \eqref{bc4na1} implies that $f$ is constant on the boundary $\partial\Omega$, which leads to the second integral in \eqref{h4na-transf} being a constant times a topological invariant
\begin{align}
  n=\int_{\partial\Omega}\epsilon^{\alpha\beta\mu\nu}\tr U^{-1}(\nabla_\mu U)U^{-1}(\nabla_\nu U)U^{-1}(\nabla_\alpha U)\,d^3\sigma_\beta.\label{}
\end{align}
The invariant $n$ is the winding number of the mapping from $\partial\Omega$ into the gauge group $G$. 
If $\partial\Omega\simeq\S^3$ and $G$ is compact, then $n$ is an integer since the corresponding homotopy group of maps $\partial\Omega\to G$ is $\pi_3(G)\simeq\Z$. 
However, if $\partial\Omega$ is not topologically equivalent to $\S^3$, than the homotopy group might be different. 

Even though the asymptotic behavior for $A$ and $f$ are fixed by \eqref{a4na-asympt} and \eqref{f4na-asympt}, the integrals in \eqref{h4na-transf} can still diverge for a non-compact $\Omega$ unless we specify appropriate asymptotic conditions for $U$.
Without loss of generality, we assume
\begin{align}
  U=\pm I+O(\norm{x}^r), \ \norm{x}\to\infty, \ r<0,\label{u4na-asympt}
\end{align}
where $I$ is the unit matrix.
The plus-minus sign here corresponds to the two possible signs of $\det{U}$, which account for the cases of proper rotations versus rotations combined with the inversion with respect to the origin.  
Vanishing of the first and second integral in \eqref{h4na-transf} implies $2p+q+r+2<0$ and $q+3r<0$, respectively.
With the above conditions satisfied, the helicity $H(f)$ is invariant with respect to both small and large gauge transformations.

We can now prove the following analog of Theorem~\ref{theorem1}.
\begin{theorem}
Any extremal of the action functional $W=\int_{\Omega}\tfrac{1}{4}\tr{F_{\mu\nu}F^{\mu\nu}} d^4x$ subject to the constraint $\epsilon^{\mu\nu\alpha\beta}\tr{F_{\mu\nu}F_{\alpha\beta}}=0$ and the boundary condition $\delta A\vert_{\partial\Omega}=0$ satisfies the covariant Beltrami equation $D^\nu F_{\mu\nu}=\tfrac{1}{2}\epsilon_{\mu\nu\alpha\beta}\lambda^\nu F^{\alpha\beta}$ for $\lambda_\mu=\nabla_\mu f$, where $f=f(x)$ is an arbitrary function.
  \label{theorem2}
\end{theorem}
\begin{proof}
Any extremal of \eqref{e4na} subject to the local constraint \eqref{epsff-na} must be an extremal of the functional $W-\tfrac{1}{2}\tilde{H}(f)$, where $f(x)$ is a space-time dependent Lagrange multiplier~\cite{vc}.
For an arbitrary variation of the gauge potential $\delta A$, we find
\begin{align}
  \delta(W-\tfrac{1}{2}\tilde{H}(f))&=\int_\Omega\tr{\delta A_\mu\bigl(D_\nu F^{\mu\nu}-\tfrac{1}{2}\epsilon^{\mu\nu\alpha\beta}(\nabla_\nu f) F_{\alpha\beta}\bigr)}\,d^4x\nn\\ &+\int_{\partial\Omega}\tr{\delta A_\mu\bigl(-F^{\mu\nu}+\tfrac{1}{2}\epsilon^{\mu\nu\alpha\beta}f F_{\alpha\beta}\bigr)} \,d^3\sigma_\nu.
  \label{var4-na}
\end{align}
Using the boundary condition
\begin{align}
  \delta A\vert_{\partial\Omega}=0,
  \label{bc4var-na}
\end{align}
we arrive at \eqref{b4na} with $\lambda_\mu=\nabla_\mu f$, which proves the theorem.
\end{proof}

We now consider the nonabelian version of the conserved Noether current \cite{Jackiw}.
Equations \eqref{functional4}, \eqref{lagrangian4}, \eqref{delta-noether4}, \eqref{noether4def}, \eqref{noether4}, \eqref{current4} become 
\begin{align}
  &W-\tfrac{1}{2}\tilde{H}(f)=\int_\Omega L(f)\, d^4 x,\label{functional4-na} \\
  &L(f)=\tfrac{1}{4}\tr{F_{\mu\nu}\bigl(F^{\mu\nu}}-\tfrac{1}{2}\epsilon^{\mu\nu\alpha\beta}fF_{\alpha\beta}\bigr),\label{lagrangian4-na} \\
  &\delta L(f)=0,\label{delta-noether4-na} \\
  &j^\mu(f)=\frac{\partial L(f)}{\partial\nabla_\mu A^a_\nu}\delta A^a_\nu, \label{noether4def-na} \\
  &j^\mu(f)=\tr{\bigl(F^{\mu\nu}-\tfrac{1}{2}\epsilon^{\mu\nu\alpha\beta}fF_{\alpha\beta}\bigr)D_\nu g}, \label{noether4-na} \\
  &k^\mu(f)=-D_\nu F^{\mu\nu}+\tfrac{1}{2}\epsilon^{\mu\nu\alpha\beta}(\nabla_\nu f)F_{\alpha\beta}. \label{current4-na}
\end{align}
Similarly to the abelian case, using now the Beltrami equation \eqref{b4na}, we find $k^\mu(f)=0$, so that the Noether current associated with the helicity, $\tfrac{1}{2}\epsilon^{\mu\nu\alpha\beta}(\nabla_\nu f)F_{\alpha\beta}$, equals the negative of the Yang-Mills current $J^\mu=-D_\nu F^{\mu\nu}$.
As expected, there is only one independent conserved current in the problem.

The Noether energy-momentum tensor is
\begin{align}
  {\theta^\mu}_\nu(f)&=\frac{\partial L(f)}{\partial(\nabla_\mu A^a_\sigma)}\nabla_\nu A^a_\sigma-{\delta^\mu}_\nu L(f),\nn\\
  &=\tr{\bigl\{(F^{\mu\sigma}-\tfrac{1}{2}f\epsilon^{\mu\sigma\alpha\beta}F_{\alpha\beta})\nabla_\nu A_\sigma-{\delta^\mu}_\nu(\tfrac{1}{4}F_{\alpha\beta}F^{\alpha\beta}-\tfrac{1}{8}\epsilon^{\alpha\beta\gamma\delta}fF_{\alpha\beta}F_{\gamma\delta})\bigr\}}
  \label{theta4-na}
\end{align}
and the corresponding energy-momentum $4$-vector is
\begin{align}
  &P_\nu(f)=\int_{\Omega'}{\theta^0}_\nu(f)\,d^3 x,\\
  &P_0(f)=-\int_{\Omega'}\tr{\bigl(\tfrac{1}{2}\norm{E}^2+\tfrac{1}{2}\norm{B}^2\bigr)}\,d^3 x+\int_{\partial\Omega'}\tr{\{n\cdot(E+fB) A_0\}} \,d^2\sigma,\\
  &P_i(f)=-\int_{\Omega'}\tr{(E\times B)_i}\,d^3 x+\int_{\partial\Omega'}\tr{\{n\cdot(E+fB)A_i\}}\,d^2\sigma.
  \label{}
\end{align}
where we assumed $\Omega=[t_1,t_2]\times\Omega'$.
We set the boundary condition
\begin{align}
  \tr{\{n\cdot(E+fB)\}}\vert_{\partial\Omega'}=0
  \label{bc-p4-na}
\end{align}
and, similarly to the abelian case, obtain the relation $P_\nu(f)=P_\nu(0)$ which is consitent with $L(f)-L(0)$ being a topological term.
Also note that although ${\theta^\mu}_\nu(f)$ is not gauge invariant, the resulting $P_\nu(f)$ is.

To prove conservation of ${\theta^\mu}_\nu(f)$, we need to use the Beltrami equation.
Indeed, in the expression 
\begin{align}
  \nabla_\mu{\theta^\mu}_\nu(f)&=\tr{\bigl\{\bigl((D_\mu F^{\mu\sigma})-\tfrac{1}{2}\epsilon^{\mu\sigma\alpha\beta}(\nabla_\mu f)F_{\alpha\beta}\bigr)\nabla_\nu A_\sigma+\tfrac{1}{8}\epsilon^{\alpha\beta\gamma\delta}(\nabla_\nu f)F_{\alpha\beta}F_{\gamma\delta}\bigr\}}
  \label{nabla-theta4-na}
\end{align}
we see that the first term on the right-hand side vanishes for any solution of \eqref{b4na} and the second term vanishes due to the constraint \eqref{epsff-na}. 
Since \eqref{epsff-na} follows from \eqref{b4na}, we conclude that the conservation equation $\nabla_\mu{\theta^\mu}_\nu(f)=0$ holds for any solution of the non-abelian Beltrami equation.

\section{Conclusions}

We have generalized the magnetic helicity and Beltrami equation to relativistic and non-abelian forms.
In the process, we discussed various interconnected features associated with these generalizations. 
In particular, we found that the helicity is related to the Chern-Simons action and can also be viewed as a constraint requiring the vanishing of a generalized instanton term.

Besides its theoretical appeal, the covariant formulation of the magnetic helicity and Beltrami equation has an experimental advantage as well.
It turns out that, for an ideal nonrelativistic plasma, charges flow until the electric field are completely shorted out.
In the relativistic case, even for an ideal plasma, however, the current flow may not be able to keep up, and so the electric fields do not necessarily always vanish.
Some possible applications of our results for the relativistic generalization of the Beltrami equation may be found for dynamos inside millisecond pulsars, pulsar and quasar atmospheres, collisions of plasma shock waves with other shocks or gas clouds and nuclear fusion via laser confinement.

The generalization to the nonabelian case is straightforward but interesting since several further features arise.
We have already briefly mentioned a few systems where our results could be useful.
They may further apply to high energy QCD collisions ranging from relativistic heavy ion collisions, where a liquid state has been suggested, to hadronization processes in high energy elementary particle collisions.
There may also be applications to the prehadronic early universe.
We hope to explore some of these topics in the future.

Explicit solutions of the covariant and non-abelian Beltrami equations are of particular interest for applications, and we will address these elsewhere.

\appendix*
\section{Main results in terms of differential forms}

It is well known that formulation of a gauge theory in terms of differential forms often adds conceptual clarity and computational convenience.
The generalized helicity and Beltrami fields are no exception in this regard.
Here we use differential forms to state our main results for the non-abelian  helicity and Beltrami fields.
These should be sufficient for the interested reader to easily fill out the remaining details and derive corresponding relations for the abelian covariant and non-covariant cases. 

Introducing the forms 
\begin{align}
  &A=A_\mu dx^\mu,\\
  &F=dA+A\wedge A,\\
  &F=\tfrac{1}{2}F_{\mu\nu}dx^\mu\wedge dx^\nu,\\
  &\lambda=\lambda_\mu dx^\mu
  \label{}
\end{align}
we rewrite \eqref{b4na-def}, \eqref{b4na}, \eqref{b4na-lambda}, \eqref{epsff-na} as
\begin{align}
  &*F\wedge D*F=0,\label{b4def-df}\\
  &D*F=\lambda\wedge F,\label{b4-df}\\
  &F\wedge d\lambda=0,\label{b4lambda-df}\\
  &\tr{(F\wedge F)}=0.\label{epsff-df}
\end{align}
Note that the left-hand side of \eqref{b4lambda-df} vanishes identically if we set $d\lambda=0$, which follows from $\lambda=df$ found in the proof of Theorem \ref{theorem2}.

The action \eqref{e4na} and helicity \eqref{h4na} become
\begin{align}
  &W=\int_\Omega\tfrac{1}{2}\tr{(F\wedge *F)},\label{e4na-diff}\\
  &H(f)=-\int_\Omega df\wedge\tr{\bigl(A\wedge F-\tfrac{1}{3}A\wedge A\wedge A\bigr)},\label{h4na-diff}
\end{align}
while \eqref{htildeh4na} and \eqref{tildeh4na} become
\begin{align}
  &H(f)=\tilde H(f)-\int_{\partial\Omega}f\tr{\bigl(A\wedge F-\tfrac{1}{3}A\wedge A\wedge A\bigr)}, \\ 
  &\tilde H(f)=\int_\Omega f\tr{(F\wedge F)}.
  \label{}
\end{align}
Under a gauge transformation
\begin{align}
  &A\mapsto U^{-1}AU+U^{-1}dU,\label{gauge4na-diff-a}\\
  &F\mapsto U^{-1}FU,\label{gauge4na-diff-f}
\end{align}
the invariance of \eqref{e4na-diff} is obvious and \eqref{h4na-transf} becomes
\begin{align}
  &H(f)\mapsto H(f)-\int_{\partial\Omega}f\tr{d\bigl(dUU^{-1}\wedge A\bigr)}-\int_{\partial\Omega}\tfrac{1}{3}f\tr{\bigl(dUU^{-1}\wedge dUU^{-1}\wedge dUU^{-1}\bigr)}.
  \label{}
\end{align}
The boundary conditions \eqref{bc4na1}, \eqref{bc4na2}, \eqref{bc4na3} become
\begin{align}
  &df\vert_{\partial\Omega}=0,\\
  &(D*F)_{\partial\Omega}=0,\\
  &(*n\wedge D*F)_{\partial\Omega}=0,
  \label{}
\end{align}
where $n=n_\mu dx^\mu$.
For the Noether current \eqref{current4-na} we have
\begin{align}
  &k(f)=k_\mu(f)dx^\mu, \\
  &k(f)=*(D*F-df\wedge F).
  \label{}
\end{align}
The Beltrami equation \eqref{b4-df} now gives $k(f)=0$.

\begin{acknowledgments}
We thank Kieth Moffatt for pointing out the early work by Carter on generalizing helicity.
We are greatful to the Isaac Newton institute for hospitality while this work was being carried out.
RVB thanks Chapman University for support.
The work of TWK was supported by U.S. DoE grant number DE-FG05-85ER40226 and Vanderbilt University College of Arts and Sciences.
\end{acknowledgments}

\end{document}